\documentclass[11pt]{article}

\usepackage{arxivbasic}

\usepackage{graphicx}
\usepackage[title]{appendix}  

\usepackage{amsmath, amsthm, amssymb, mathdots} 
\allowdisplaybreaks 

\newtheorem{theorem}{Theorem}

\newtheorem{corollary}{Corollary} 

\theoremstyle{definition}

\usepackage{dsfont} 

\usepackage{enumitem}

\usepackage[font=small,labelfont=bf]{caption}

\usepackage[mathscr]{eucal}

\usepackage{microtype}
 \DisableLigatures{encoding = *, family = * }

\usepackage{hyperref}
\hypersetup{colorlinks=true, citecolor=black, linkcolor=black, urlcolor=black}

\usepackage[
	backend=bibtex,
	style=authoryear,
	citestyle=authoryear-comp,
	natbib=true,
	doi=true,isbn=false,url=false ]{biblatex}
\begin{document}

\title{Finding Reproduction Numbers for\\ Epidemic Models \& Predator\textendash{}Prey Models of \\ Arbitrary Finite Dimension Using \\ The Generalized Linear Chain Trick}

\author{
    Hurtado, Paul J. \\
	University of Nevada, Reno \\
	ORCID: 0000-0002-8499-5986 \\
	\texttt{phurtado@unr.edu} 
	\And
	Richards, Cameron\\
	University of Nevada, Reno\\ 
	ORCID: 0000-0002-1620-9998
}


\maketitle

\begin{abstract} Reproduction numbers, like the basic reproduction number $\mathcal{R}_0$, play an important role in the analysis and application of dynamic models, including contagion models and ecological population models. One difficulty in deriving these quantities is that they must be computed on a model-by-model basis, since it is typically impractical to obtain general reproduction number expressions applicable to a family of related models, especially if these are of different dimensions. For example, this is typically the case for SIR-type infectious disease models derived using the linear chain trick (LCT). Here we show how to find general reproduction number expressions for such models families (which vary in their number of state variables) using the next generation operator approach in conjunction with the generalized linear chain trick (GLCT). We further show how the GLCT enables modelers to draw insights from these results by leveraging theory and intuition from continuous time Markov chains (CTMCs) and their absorption time distributions (i.e., phase-type probability distributions). To do this, we first review the GLCT and other connections between mean-field ODE model assumptions, CTMCs, and phase-type distributions. We then apply this technique to find reproduction numbers for two sets of models: a family of generalized SEIRS models of arbitrary finite dimension, and a generalized family of finite dimensional predator-prey (Rosenzweig-MacArthur type) models. These results highlight the utility of the GLCT for the derivation and analysis of mean field ODE models, especially when used in conjunction with theory from CTMCs and their associated phase-type distributions.	
\end{abstract}

~\\

\clearpage
\tableofcontents
\clearpage


\section{Introduction} \label{sec:intro} 

The basic reproduction number $\mathcal{R}_0$ is perhaps the most well known threshold quantity derived from epidemic disease models \citep{Diekmann1990, Hethcote2000, Hyman2000, Hyman2005, vdDW2002, Heffernan2005, Roberts2012, Heesterbeek2002, MathEpi2008, BrauerCCC2011, Diekmann2000, Kermack1927, Kermack1932, Kermack1933, Kermack1991a, Kermack1991b, Kermack1991c,Diekmann2009,Dietz1993}. Along with its time-varying counterpart, the effective reproduction number $\mathcal{R}_t$ \citep[also called the replacement number;][]{Hethcote1976,Hethcote2000}, $\mathcal{R}_0$ plays a central role in how we understand and attempt to control the spread of infectious diseases, since these quantities summarize how the overall transmission process reflects the combined impact of various biological processes, e.g., host susceptibility, infectiousness, recovery, between-host contact processes, and control measures such as vaccination and quarantine. Such reproduction numbers can also be derived for models of other kinds of contagion and other biological populations, e.g., multispecies ecological models \citep{Hilker2008,Hurtado2014,Roberts2012,Duffy2016}, cell-level models of cancer \citep[e.g.,][]{Eftimie2011}, and viral infection \citep[e.g.,][]{Hews2009}. 

Reproduction number expressions are, however, only as good as the model assumptions from which they are derived. For example, the importance of incorporating non-exponentially distributed latent and infectious periods in SEIR-type models (see section \ref{sec:SEIRmodels} below) is well known \citep{Wearing2005}, and so is the potential importance of incorporating age structure and maturation delays into ecological population models \citep{Xia2009,Smith1974,Wilmers2007,Cushing1982,Hastings1983,Wang2019,Levine1983}.

The standard linear chain trick \citep[LCT;][]{Smith2010,Hurtado2019,MacDonald1978,Metz1991} has been widely used for decades as a way of replacing the assumption of exponentially distributed dwell times with Erlang distributions (i.e., gamma distributions with integer shape parameters) in ODE models. This technique yields a system of ODEs whose dimension depends on Erlang distribution shape parameter value(s), which determine(s) the number of model state variables and their corresponding equations. Deriving a new model using the LCT, therefore, defines a countably large family of new models, each of a different (finite) dimension. One challenge with this approach is that it has proven difficult to derive a single, general formula for the basic reproduction numbers for such model families. Instead, reproduction number expressions are typically derived on a case-by-case basis using a single, fixed shape parameter value (although, see \citet{Hyman2005,Bonzi2010}). Thus, there is a need for improving the available techniques for finding reproduction number expressions so that they are more suitable for such applications. As we show below, the generalized linear chain trick \citep[GLCT;][]{Hurtado2019} can help meet this challenge. 

To illustrate how, here we give a brief overview of the GLCT and related concepts from Markov chain theory, and then introduce and analyze two families of models: a generalized SEIRS model family that encompasses multiple SEIRS-type models as special cases, and a family of generalized Rosenzweig-MacArthur type predator-prey (consumer-resource) models with a stage-structured predator population. We show how the next generation operator approach \citep{vdDW2002,Diekmann1990,Diekmann2000,Heesterbeek2002,Diekmann2009,Roberts2012} can be used in conjunction with the GLCT to derive a general expression for the predator reproduction number $\mathcal{R}_{pred}$ in the predator-prey models, and to derive a general expression for the basic reproduction number $\mathcal{R}_0$ that holds for all instances of the general SEIRS model. Importantly, these results hold regardless of the dimension of those models so long as they remain finite. 

Key to the success of this approach is that 1) we conduct this analysis on a matrix-vector form of the model, as in \citet{Hyman2005}, and 2) because we use GLCT-based matrix-vector model formulations we can interpret the resulting reproduction number expressions through the lens of continuous time Markov chains (CTMCs), related stochastic processes, and \textit{phase-type distributions} \textendash{} the broad family of hitting time (or absorption time) distributions for absorbing CTMCs, which includes exponential, generalized Erlang, and Coxian distributions \citep{Bladt2017,Bladt2017ch3,Reinecke2012a,BuTools2}. 

The sections below are organized as follows. We first review the GLCT, LCT, and other connections between mean field ODE models, Markov chain theory, and phase-type distributions. We then review some standard SEIRS models and their basic reproduction numbers ($\mathcal{R}_0$), and introduce our generalized SEIRS and predator-prey models before deriving general reproduction number expressions for both.

\section{Background}

\subsection{CTMCs \& Phase-type distributions} 

In this section, we give a brief overview of the phase-type family of univariate, matrix exponential probability distributions, and their connection to continuous time Markov chains (CTMCs)\footnote{For readers only familiar with discrete time Markov chains, CTMCs are very similar, except that only state transitions to a different state are considered (so, transition probabilities $p_{ii}=0$, $i=1\ldots,k$), and rather than state transitions occurring after a fixed time step, they occur after an exponentially distributed duration of time (i.e., an exponentially distributed \textit{dwell time}), where the time spent in the $i^\text{th}$ state is exponential with its own rate $r_i>0$. CTMCs can be parameterized with a transition probability matrix, just like a discrete time Markov chain, and also a rate vector of $r_i$ values. These quantities are more commonly combined into a \textit{transition rate matrix}, like eq. \eqref{eq:ratematrix}, with $-r_i$ (net loss rate from the $i^\text{th}$ state) along the diagonal, and off-diagonal entries (row $i$ column $j$) of the form $r_i\,p_{ij}$ where $p_{ij}$ is the transition probability from the $i^\text{th}$ state to the $j^\text{th}$ state, and thus  $r_i\,p_{ij}$ is the \textit{rate} at which individuals move from the $i^\text{th}$ state into the $j^\text{th}$ state.\label{footnote:ratematrix}}, based on \citet{Hurtado2019} and \citet{HurtadoRichards2021,HurtadoRichards2020b}. Some familiar examples of distributions in the phase-type family are the exponential distribution, Erlang distribution (i.e., gamma distributions with integer shape parameters $k$), hypoexponential (or generalized Erlang) distribution, hyper-Erlang (finite Erlang mixture) distribution, and the Coxian distribution. Various statistical tools exist for fitting these phase-type distributions to data \citep[e.g.,][]{BuTools2,BuToolswww}. For more details, see \citet{Bladt2017,Reinecke2012a, Reinecke2012b, Horvath2012, Horvath2016, Altiok1985}.

Phase-type distributions describe the time it takes to first reach an absorbing state in a CTMC. A given phase-type distribution is parameterized by a matrix $\mathbf{A}$ and column vector $\boldsymbol{\alpha}$, as detailed in eqs. \eqref{eq:pt} below. These quantities define a corresponding CTMC with $k$ transient states and one absorbing state, as follows. This CTMC has an initial state distribution vector $[\alpha_1,\ldots,\alpha_k,\alpha_*]^\text{T}$ determined by \begin{equation}\label{eq:alpha}\boldsymbol{\alpha}=[\alpha_1,\ldots,\alpha_k]^\text{T}\end{equation} 
(the initial distribution across transient states) which determines the fraction that immediately enter the absorbing state, $\alpha_* = 1- \sum_{i=1}^k\alpha_i$.  The $(k+1) \times (k+1)$ \textit{transition rate matrix}\footnotemark[\value{footnote}] is given by  \begin{equation}\label{eq:ratematrix}\begin{bmatrix} \mathbf{A} & \mathbf{a} \\ \mathbf{0} & 0 \end{bmatrix}\end{equation} where $\mathbf{A}$ is the $k\times k$ transient state block of the rate matrix, $\mathbf{a}$ is a length $k$ column vector, $\mathbf{0}$ is a length $k$ row vector of zeros.  

The negative values along the diagonal\footnotemark[\value{footnote}] of $\mathbf{A}$ are the loss rates from the corresponding transient states to any other state, and the non-negative off diagonal entries\footnotemark[\value{footnote}] of $\mathbf{A}$ are the rates of influx into each transient state from the other transient states \citep{Bladt2017,Hurtado2019, HurtadoRichards2021, HurtadoRichards2020b}. The elements of vector $\mathbf{a}$ in eq. \eqref{eq:ratematrix} are the loss rates from each transient state to the absorbing state. Since the rows in a transition rate matrix like eq. \eqref{eq:ratematrix} must sum to zero, this vector can be written in terms of $\mathbf{A}$ as \begin{equation}\mathbf{a}=-\mathbf{A\,1}\label{eq:a}\end{equation} 
where (here, and below) $\mathbf{1}$ denotes an appropriately long column vector of ones. Therefore, the CTMC corresponding to a given phase-type distribution is fully determined by $\boldsymbol{\alpha}$ and $\mathbf{A}$.

Phase-type distributions have a probability density function $f(t)$, cumulative density function (CDF) $F(t)$, Laplace-Stieltjes Transform of the CDF $\mathcal{L}(s)$, and $j^\text{th}$ moment  $E(Y^j)$ given, respectively, by eqs. \eqref{eq:pt} below, which are adapted from \citet{Reinecke2012a}: 

\begin{subequations} \label{eq:pt} \begin{align}
		f(t) =&\; \boldsymbol{\alpha}^\text{T}\,e^{\mathbf{A}t}\,(-\mathbf{A}\mathbf{1}) \\
		F(t) =&\; 1 - \boldsymbol{\alpha}^\text{T}\,e^{\mathbf{A}t}\,\mathbf{1} \\
		\mathcal{L}(s) =&\; \alpha_* + \boldsymbol\alpha^\text{T}\,(s\mathbf{I}-\mathbf{A})^{-1}(-\mathbf{A\,1})\\
		E(Y^j)=&\; j!\,\boldsymbol{\alpha}^\text{T}\,(-\mathbf{A})^{-j}\mathbf{1}. 
\end{align} \end{subequations}

Superscript $^\text{T}$ denotes the matrix transpose. As above, $\alpha_*$ is the probability of starting the CTMC in the absorbing state (typically $\alpha_*=0$), $\mathbf{1}$ is a column vector of ones, and $\mathbf{I}$ is the identity matrix, where each has appropriate dimensions given $\mathbf{A}$. If a phase-type distribution with parameters $\boldsymbol{\alpha}$ and $\mathbf{A}$ has $\alpha_*>0$, then it can be thought of as the zero-inflated mixture distribution of the phase-type distribution with parameters $\boldsymbol{\alpha}/(1-\alpha_*)$ and $\mathbf{A}$, and a Dirac delta distribution (point mass w.p. 1) at 0 with respective mixing probabilities $1-\alpha_*$ and $\alpha_*$. Note that a given phase-type distribution does not necessarily have a unique parameterization.

One important family of phase-type distributions are the Erlang distributions. These are the gamma distributions with integer shape parameters. More specifically, the sum of $k$ \textit{i.i.d.} exponential random variables with rate $r$ is Erlang distributed with rate $r$ and shape $k$. Erlang distributions can also be parameterized by their mean $\tau$ and variance $\sigma^2$ (or coefficient of variation $c_v=\sigma/\tau$): \begin{equation}
\tau=\frac{k}r,\quad \sigma^2 = \frac{k}{r^2}, \quad c_v = \frac{1}{\sqrt{k}}, \quad \text{and thus,} \quad k=\frac{\tau^2}{\sigma^2}=\frac{1}{c_v^{\;2}},\quad \text{and} \quad r=\frac{\tau}{\sigma^2}=\frac{1}{c_v^{\;2}\tau}=\frac{k}{\tau}.
\end{equation} 

The generalized Erlang distributions (also called hypoexponential distributions) are equivalent to the sums of $k$ independent exponential distributions, which may have distinct rates $r_i$, $i=1,\ldots,k$. 

Another important family of phase-type distributions are the Coxian distributions.  These are the absorption time distributions for CTMCs where each of $k$ transient states has its own rate $r_i$ (similar to generalized Erlang distributions) but for each state there is some probability $p_i$ of entering the next state in the ``chain'', or (with probability $1-p_i$) transitioning straight to the absorbing state (see Appendix \ref{A:minExpErlang} for an example). Phase-type distributions can be classified as \textit{acyclic} (transient states cannot be revisited once left) and \textit{cyclic} (one or more transient states can be revisited multiple times due to cycles in the transition rate matrix), and all acyclic phase-type distributions have Coxian representations \citep{Cumani1982,OCinneide1991,OCinneide1993}.

\subsubsection{Some useful properties of phase-type distributions and CTMCs}

\textbf{Closure.} The phase-type family of distributions is closed under various operations including addition (convolution), minimum, maximum, and finite mixtures \citep{Bladt2017,Neuts1981}.

\textbf{Minimum of Phase-Type Random Variables.} As mentioned above, the minimum of two independent phase-type random variables, parameterized by $\boldsymbol{\alpha_i}$, $\mathbf{A_i}$ ($i=1,2$), is also phase-type distributed. More specifically, its parameters are given by the following Kronecker product and sum \citep{Bladt2017}:  \begin{subequations}\label{eq:minPT}\begin{align} \boldsymbol{\alpha_\text{min}}=&\;\boldsymbol{\alpha}_1 \otimes \boldsymbol{\alpha}_2\\ \mathbf{A_\text{min}} =&\; \mathbf{A_1}\oplus \mathbf{A_2}.\end{align}\end{subequations} 

This statement extends to the minimum of more than two phase-type distributions by the properties of Kronecker products and sums.

\textbf{Expected Rewards.} The following is a property of the \textit{reward process} associated with a given CTMC. Specifically, suppose that $g_i$ is the \textit{reward rate}\footnote{This reward rate can be thought of as a mean rate of reward accrual in a more general renewal reward process context.} associated with the $i^\text{th}$ transient state in an absorbing CTMC, where the reward amount accrued while spending a duration of time $T_i$ in the $i^\text{th}$ state is $g_i\,T_i$. Let $W$ be the total reward accrued prior to reaching the absorbing state, and let $\mathbf{g}$ be the column vector of reward rates for the transient states. The expected reward accrued prior to reaching an absorbing state is (see Appendix \ref{A:reward} for a proof) \begin{equation} \label{eq:Ereward}	E(W)=\; \boldsymbol{\alpha}^\text{T}\,(-\mathbf{A}^{-1})\mathbf{g}.  \end{equation}

Recalling the equations above, let $\mathbf{U}=-\mathbf{A}^{-1}$ (this is called the \textit{Green} matrix, and is analogous to the \textit{fundamental matrix} in discrete time Markov chains). The entries $u_{ij}$ represent the expected time spent in the $j^\text{th}$ state prior to reaching the absorbing state, given that the initial state was the $i^\text{th}$ state \citep{Bladt2017}. 

Note that the expected time to absorption can be obtained from eq. \eqref{eq:Ereward} if the reward rates are all $g_i=1$ (cf. eq. \eqref{eq:pt}).

\textbf{CTMC Absorption Probabilities.} It is useful to know how individuals are distributed across multiple absorbing states upon leaving the transient states when an absorbing CTMC has more than one absorbing state. Obtaining the overall distribution across absorbing states can be achieved as follows, using standard Markov chain theory on the embedded jump process \citep{Resnick2002,Bailey1990} corresponding to the given phase-type distribution and additional assumptions related to the multiple absorbing states.  

Assume a state X is partitioned into $k$ sub-states X$_i$ and these constitute the transient states associated with a phase-type distribution with parameters $\boldsymbol{\alpha}\in\mathbb{R}^k$ and $\mathbf{A}\in\mathbb{R}^{k\times k}$. Assume the corresponding CTMC also has $m>1$ absorbing states. Suppose $\mathbf{C}$ is a $k\times m$ matrix where the row $i$ and column $j$ entry of $\mathbf{C}$ is the probability that an individual enters the $j^\text{th}$ absorbing state given that it left the $i^\text{th}$ transient state (the rows of $\mathbf{C}$ sum to 1). One can compute a $1 \times m$ vector of absorption probabilities $\boldsymbol{\pi}$, where $\pi_i$ is the expected fraction of individuals that ultimately enter the $i^\text{th}$ absorbing state, as follows.

To begin, $\boldsymbol{\pi}$ depends only on the transition probabilities, not the rates. Thus, we must first obtain the \textit{transition probability matrix} $\mathbf{P}$ for the embedded jump process implicit in a transition rate matrix like eq. \eqref{eq:ratematrix} for a phase-type distribution with parameters $\mathbf{A}$ and $\boldsymbol{\alpha}$. Since this matrix has the form \begin{equation}  \mathbf{P} =\; \begin{bmatrix} \mathbf{Q} & \mathbf{R} \\ \mathbf{0} & \mathbf{I} \end{bmatrix}  \end{equation}
it suffices to find $\mathbf{Q}$ and $\mathbf{R}$. 

The vector of exponential dwell-time rates for the transient states is $\mathbf{r}=-\text{diag}(\mathbf{A})$, where vector $\text{diag}(\mathbf{M})$ denotes the diagonal entries of matrix $\mathbf{M}$. 

Define matrix $\mathbf{Q}$ by taking matrix $\mathbf{A}+\mathbf{D_r}$ \textendash{} where $\mathbf{D_v}$ denotes a diagonal matrix with entries of vector $\mathbf{v}$ along its diagonal \textendash{} and dividing each ($i^\text{th}$) row by the corresponding rate $r_i$. This yields \begin{equation}\label{eq:AQ} \mathbf{Q} =\; \mathbf{D_r}^{-1} \mathbf{A} +  \mathbf{I}, \quad \text{thus} \quad -\mathbf{A}=\;\mathbf{D_r}\big(\mathbf{I} - \mathbf{Q}\big) \qquad \text{and} \qquad \mathbf{D_r}^{-1}\mathbf{A} =\;  \mathbf{Q}  - \mathbf{I}.\end{equation}

$\mathbf{R}$ has the same dimensions as $\mathbf{C}$, and can be determined from $\mathbf{C}$ as follows. The loss rate vector $\mathbf{a}=-\mathbf{A1}$ in eqs. \eqref{eq:ratematrix} and \eqref{eq:a} can be divided (element-wise) by the rate vector $\mathbf{r}$ which yields the probability vector \begin{equation}
\label{eq:q} \mathbf{q}=\mathbf{D_r}^{-1}(-\mathbf{A1}).
\end{equation} The entries $q_i=a_i/r_i$ are the probabilities of entering an absorbing state upon leaving the $i^\text{th}$ transient state. Since the $i^\text{th}$ row of $\mathbf{C}$ is the conditional distribution of individuals across the absorbing states coming from the $i^\text{th}$ transient state, the $i^\text{th}$ row of $\mathbf{R}$ is $q_i$ times the $i^\text{th}$ row of $\mathbf{C}$. \begin{equation}
\mathbf{R} =\; \mathbf{D_q}\,\mathbf{C} =\; \mathbf{D_r}^{-1}\mathbf{D_a}\,\mathbf{C}.
\end{equation}

The desired absorption probabilities can be calculated from $\mathbf{Q}$ and $\mathbf{R}$ as one would for a discrete time Markov chain \citep{Resnick2002,Bladt2017}. The probabilities of hitting the $j^\text{th}$ absorbing state, given the initial state was the $i^\text{th}$ transient state, are the $ij$ entries in the matrix \begin{equation} \label{eq:absprob} \big(\mathbf{I} - \mathbf{Q}\big)^{-1} \,\mathbf{R} = \big(-\mathbf{D_r}^{-1}\mathbf{A} \big)^{-1} \,\mathbf{R} =\; (-\mathbf{A}^{-1} )\mathbf{D_r}\,\mathbf{R} =\; (-\mathbf{A}^{-1} ) \mathbf{D_a}\,\mathbf{C}.\end{equation}

Finally, averaging the columns in the above matrix across the initial distribution vector $\boldsymbol\alpha$ (assuming $\sum_{i=1}^k\alpha_i=1$) yields the absorption probabilities \begin{equation}
\label{eq:absprobavg} \boldsymbol{\pi} =\; \boldsymbol{\alpha}^\text{T} \; \big(\mathbf{I} - \mathbf{Q}\big)^{-1} \,\mathbf{R}.\end{equation}

For additional properties of phase-type distributions, see \citet{Bladt2017,Reinecke2012a, Reinecke2012b, Horvath2012, Horvath2016, Altiok1985} and  references cited in \citet{Hurtado2019}. For additional properties of CTMCs and associated processes, consult a standard text on stochastic processes, e.g., \citet{Resnick2002,Bailey1990,Karlin1975,Taylor1998}.

\subsection{Generalized linear chain trick (GLCT)}

The following overview of the Generalized Linear Chain Trick (GLCT) is adapted from the more detailed presentation of the GLCT found in \citet{Hurtado2019} and from \citet{HurtadoRichards2021, HurtadoRichards2020b}.  As we show below, ODE models derived using the GLCT can be used to prove general results for models that could otherwise be obtained using the standard Linear Chain Trick (LCT). This is partly the result of the equations being in a specific matrix-vector form \textendash{} which can be analyzed without the constraint of a fixed model dimension \textendash{} but also because the analysis of such models can give rise to quantities associated with other related phase-type distributions and related quantities, and it can be useful to recognize them as such.

The GLCT enables modelers to interpret many existing ODE models through the lens of Markov chain theory, and to derive new mean field ODE models relatively quickly by bypassing the need to explicitly derive those ODEs using mean field integral equations and their derivatives \citep{Hurtado2019, HurtadoRichards2020b}. Such derivations can be obtained using the GLCT to generalize an existing model, or to derive a new ODE model from first principles as discussed in \citet{Hurtado2019}. Here we take the first approach, as detailed in \citet{HurtadoRichards2020b}, which can be summarized as follows.

Just as the LCT is used to introduce Erlang distributed dwell times into an existing ODE model, the GLCT can be used to take an existing (e.g., ODE or DDE) model and modify its assumptions to introduce phase-type distributed dwell times resulting in a new ODE model. This can most easily be done by first applying the standard LCT, then writing that new set of ODEs in matrix-vector form \textit{a la} Theorem \ref{th:glct} below to apply the GLCT for phase-type distributions (see \citet{HurtadoRichards2020b} for examples).  

Theorem \ref{th:glct} below is a re-statement of the GLCT for phase-type distributions (Corollary 2) in \citet{Hurtado2019}, as it is stated in \citet{HurtadoRichards2020b} using the above notation.  

\begin{theorem}[\textbf{GLCT for Phase-Type Distributions}] \label{th:glct}
	Assume individuals enter a state (call it state X) at rate $\mathcal{I}(t)\in\mathbb{R}$ and that the distribution of time spent in state X follows a continuous phase-type distribution given by the length $k$ initial probability vector $\boldsymbol\alpha$ and the $k\times k$ matrix $\mathbf{A}$. Then partitioning X into $k$ sub-states X$_i$, and denoting the corresponding amount of individuals in state X$_i$ at time $t$ by $x_i(t)$, then the mean field equations for these sub-states $x_i$ are given by \begin{equation} \frac{d}{dt}\mathbf{x}(t)=\boldsymbol\alpha\,\mathcal{I}(t) + \mathbf{A}^\text{T}\,\mathbf{x}(t) \label{eq:GLCT}\end{equation} where the rate of individuals leaving each of these sub-states of X is given by the vector $(-\mathbf{A\,1})\circ\mathbf{x}$, where $\circ$ is the Hadamard (element-wise) product of the two vectors, and thus the total rate of individuals leaving state X is given by the sum of those terms, i.e., $(-\mathbf{A\,1})^\text{T}\mathbf{x}=-\mathbf{1}^\text{T}\mathbf{A}^\text{T}\mathbf{x}$.
\end{theorem}

The standard linear chain trick (LCT) is a special case of Theorem \ref{th:glct}. For completeness, it is provided below as stated in \citet{HurtadoRichards2020b}. The LCT is stated here for generalized Erlang (hypoexponential) distributions (i.e., the distribution arising from the sum of $k$ independent exponentially distributed random variables, with rates $r_i$, $i=1,\ldots,k$). 

\begin{corollary}[\textbf{Linear Chain Trick for Erlang and Hypoexponential Distributions}] \label{th:lct}
	~\\
	Consider the GLCT above (Theorem \ref{th:glct}). Assume that the dwell-time distribution is a generalized Erlang (hypoexponential) distribution with rates $\mathbf{r}=[r_1,\;r_2,\;\ldots\;,r_k]^\text{T}$, where $r_i>0$, or an Erlang distribution with rate $r$ (all rates $r_i=r$) and shape $k$ (or if written in terms of shape $k$ and mean $\tau=k/r$, use $r=k/\tau$). Then the corresponding mean field equations are 
	\begin{equation} \label{eq:LCT}\begin{split} 
	\frac{dx_1}{dt} =&\; \mathcal{I}(t) - r_1\,x_1 \\
	\frac{dx_2}{dt} =&\;  r_1\,x_1 - r_2\,x_2 \\
	&\vdots \\
	\frac{dx_k}{dt} =&\; r_{k-1}\,x_{k-1} - r_{k}\,x_{k}.
	\end{split}\end{equation}	
\end{corollary}

\begin{proof} The phase-type distribution formulation of a generalized Erlang distribution with rates $r_i>0$, $i=1, \ldots, k$ is given by eqs. \eqref{eq:alphaA}. Substituting these into eq. \eqref{eq:GLCT} yields equations \eqref{eq:LCT}. \begin{equation} \label{eq:alphaA} \boldsymbol\alpha = \begin{bmatrix} 1 \\ 0 \\ \vdots \\ 0  \end{bmatrix}  \qquad \text{and} \qquad \mathbf{A} = \begin{bmatrix} -r_1 & r_1 & 0 & \cdots & 0 \\
0 & -r_2 & r_2  & \ddots & 0 \\
\vdots & \ddots & \ddots  & \ddots & \ddots \\
0 & 0 & \ddots &  -r_{k-1} & r_{k-1} \\
0 & 0 & \cdots &  0 & -r_k
\end{bmatrix}. \end{equation}  \end{proof}

See \citet{Hurtado2019} for additional results that further clarify the link between stochastic model assumptions and mean field ODE model equations, including a more general statement of the GLCT.

In addition to the statements above, the following intuition \textendash{} regarding individuals transitioning from one state into one of multiple other states \textendash{} is useful for understanding the links between underlying stochastic model assumptions and the corresponding mean field ODE model structure. 

Consider the following pair of simple scenarios, which yield the same mean field equations according to Theorem \ref{th:tomulti} below. 

Suppose the time an individual spends in a given state is assumed to be the minimum of $k$ exponentially distributed random variables (event times) with respective rates $r_i>0$, $i=1,\ldots,k$. The individual leaves their current state once the first of those events occurs. Further assume that if it is the $i^\text{th}$ of these events that occurs first, then the individual transitions to the $i^\text{th}$ recipient state with probability 1. Recall that the minimum of multiple independent exponentially\footnote{More generally, in the context of nonhomogeneous Poisson processes, the minimum of two $1^\text{st}$ event times corresponding to two independent Poisson processes, with rates $r_1(t)$ and $r_2(t)$, can be thought of as the $1^\text{st}$ event time under a Poisson process with rate $r(t)=r_1(t) + r_2(t)$. See \citet{Hurtado2019} for details.} distributed random variables (each with rate $r_i>0$, $i=1,\ldots,n$) is itself exponentially distributed with rate $r=r_1+\cdots+r_n$.

It follows that the mean field ODE terms that correspond to the scenario just described are equivalent to the mean field ODE terms obtained by assuming that the dwell time in the focal state is exponentially distributed with rate $r=r_1 + \cdots + r_n$, and that upon leaving that state individuals are distributed across the $n$ recipient states with probabilities $r_i/r$. 

The following theorem (a special case of Theorem 7 in \citet{Hurtado2019}) gives a more formal, and more general, statement of the mean field equivalence of the two scenarios above.

\begin{theorem}[\textbf{Mean field equivalence of proportional outputs \& competing exponential event times}]\label{th:tomulti} Suppose state X has a dwell time given by random variable $T=\min_i T_i$, where each $T_i$ is exponentially distributed with rate $r_i$, $i=1,\ldots,n$ and individuals transition to one of $m$ states Y$_\ell$, $\ell=1,\ldots,m$, with probability $p_{i\ell}(T)$ when $T=T_i$. The corresponding mean field model is equivalent to having instead assumed that the X dwell time is exponentially distributed with rate $r=\sum_{i=1}^n r_i$, and the transition probability vector for leaving X and entering one of the states Y$_\ell$ is given by $p_\ell=\sum_{i=1}^n p_{i\ell}\,r_i/r$.
\end{theorem}

\textbf{Example:} Consider the well known SIR model. The net loss rate from the infected state is often written $-(\gamma + \mu + \nu)I$ where the recovery rate is $\gamma$, the baseline mortality rate is $\mu$, and the disease-induced mortality rate is $\nu$. Accordingly, the fraction of those individuals leaving the infected state and entering the recovered state is the product of that total loss rate and recovery fraction: $\frac{\gamma}{\gamma + \mu + \nu}$ and thus the rate of individuals entering the recovered state is $$(\gamma + \mu + \nu)I\;\frac{\gamma}{\gamma + \mu + \nu}= \gamma\,I.$$

\subsection{SEIR models} \label{sec:SEIRmodels}

The well known SEIR model with mass action transmission, a fixed population size, and no births or deaths is given by eqs. \eqref{eq:SEIR} below \citep{Kermack1927,Kermack1991a,AndersonMay1992,BrauerCCC2011}. A more general SEIR model was introduced in \citet{HurtadoRichards2021}, which has phase-type distributed latent and infectious periods. Below, we further extend this model to include births and deaths, as well as waning immunity, where that duration of immunity is also phase-type distributed. Our goal is to derive an $\mathcal{R}_0$ expression for this general model, but first, we review these simpler models and their basic reproduction numbers.

Consider a population of size $N=S+E+I+R$, where $S$ is the number of susceptible individuals in the population, $E$ the number of individuals with latent infections (i.e., not yet symptomatic or contagious), $I$ the number of contagious infected individuals, and $R$ is the number of recovered individuals (for convenience, we will refer to these states with non-italic letters S, E, I, and R). 
\begin{subequations}  \label{eq:SEIR}  \begin{align}
	\frac{dS}{dt} =& \; -\beta\,S\,I \label{eq:SEIRa} \\
	\frac{dE}{dt} =& \; \beta\,S\,I - r_E\,E \label{eq:SEIRb}\\
	\frac{dI}{dt} =& \; r_E\,E - r_I\,I \label{eq:SEIRc}\\
	\frac{dR}{dt} =& \; r_I\,I \label{eq:SEIRd} 
	\end{align}
\end{subequations}

The model eqs. \eqref{eq:SEIR} can be interpreted as the mean field model for a continuous time stochastic SEIR model in which the time individuals spend in states E and I follow exponential distributions with respective rates $r_E$ and $r_I$, or equivalently, with respective means $\tau_E=1/r_E$ and $\tau_I=1/r_I$. 

The SEIRS model given by eqs. \eqref{eq:SEIRS} extends the above model to include births, deaths, and waning immunity. This model has a single locally asymptotically stable disease free equilibrium. This model assumes a constant birth rate $\Lambda$ and per-capita mortality rates $\mu$ which do not differ among the different disease states, thus $dN/dt=\Lambda - \mu\,N$. Upon recovery, immunity lasts for a period of time before individuals return to the susceptible class at rate $\epsilon$. Again, we can interpret these parameters as the mortality rate $\mu$ or mean lifetime $\tau_N=1/\mu$, and waning immunity loss rate $\epsilon$ or the mean duration of immunity $\tau_R=1/\epsilon$.

\begin{subequations}  \label{eq:SEIRS} \begin{align}
	\frac{dS}{dt} =& \; {\Lambda - \mu\,S} - \beta\,S\,I {+ \epsilon\,R} \\
	\frac{dE}{dt} =& \; \beta\,S\,I - r_E\,E {- \mu\,E} \\
	\frac{dI}{dt} =& \; r_E\,E - r_I\,I {- \mu\,I}\\
	\frac{dR}{dt} =& \; r_I\,I {- \mu\,R} {- \epsilon\,R}
	\end{align}
\end{subequations}

This SEIRS model can be further extended using the LCT to instead assume Erlang distributed latent and infectious periods\footnote{A simple generalization of this model would be to allow different rates for each sub-state of $E$ and $I$. Then the dwell times would follow the sum of independent but non-identically distributed exponentials each with their own rates $r_i$. The distribution of such sums are known as the \textit{hypoexponential} or \textit{generalized Erlang} distributions.}, with the same means $\tau_E$ and $\tau_I$ as the models above, but with coefficients of variation $c_{vE}$ and $c_{vI}$ chosen so that $k_E=1/\sqrt{c_{vE}}$, $k_I=1/\sqrt{c_{vI}}$ are integers\footnote{Alternatively, one could smoothly specify the $c_v$ using the mixture distribution approximation in the appendix of \citet{Hurtado2019}}. 

Equations for the resulting SIERS model with Erlang distributed latent and infectious periods are
\begin{subequations} \label{eq:SEIRSlct} \begin{align}
	\frac{dS}{dt} =& \; \Lambda -\beta\,I\,S - \mu\,S + \epsilon\,R \\ 
	\frac{dE_1}{dt} =& \; \beta\,S\,I - \frac{k_E}{\tau_E}E_1  - \mu \,E_1 \\
	\frac{dE_j}{dt} =& \; \frac{k_E}{\tau_E}E_{j-1} - \frac{k_E}{\tau_E}E_j  - \mu\, E_j, \qquad j=2,\ldots,k_E-1 \\
	\frac{dE_{k_E}}{dt} =& \; \frac{k_E}{\tau_E}E_{k_E-1} - \frac{k_E}{\tau_E}E_{k_E}  - \mu\, E_{k_E} \\ 
	\frac{dI_1}{dt} =& \; \frac{k_E}{\tau_E}E_{k_E} - \frac{k_I}{\tau_I}I_1  - \mu\, I_1\\
	\frac{dI_i}{dt} =& \; \frac{k_I}{\tau_I}I_{i-1} - \frac{k_I}{\tau_I}I_i   - \mu\, I_i, \qquad i=2,\ldots,k_I-1 \\
	\frac{dI_{k_I}}{dt} =& \; \frac{k_I}{\tau_I}I_{k_I-1} - \frac{k_I}{\tau_I}I_{k_I}   - \mu \,I_{k_I} \\ 
	\frac{dR}{dt} =& \; \frac{k_I}{\tau_I}I_{k_I}  - \mu\, R - \epsilon\,R.
	\end{align} \end{subequations}

Let us next review the derivation and interpretation of $\mathcal{R}_0$ for the three models above.

\subsubsection{Basic reproduction numbers ($\mathcal{R}_0$) for SEIR and SEIRS models} \label{sec:R0}

Using the method in \citet{vdDW2002}, the $\mathcal{R}_0$ expression for eqs. \eqref{eq:SEIR} is 

\begin{equation}\label{eq:R0.SEIR}
\mathcal{R}_0 = \frac{\beta\,S_0}{r_I} = \!\!\!\!\underbrace{\beta\,S_0}_{\substack{\text{rate of new}\\\substack{\text{infections}}}}\!\!\!\!\!\!\!\!\!\!\!\!\!\!\overbrace{\frac{1}{r_I}.}^{\substack{\text{mean infectious}\\\substack{\text{period}}}}
\end{equation}

Observe that we can factor $\mathcal{R}_0$ into the product of a (per infected individual, per unit time) rate of new infections times the mean duration of infectiousness. We therefore interpret $\mathcal{R}_0$ as the expected number of new infectious cases per infectious individual added to a population at the disease free equilibrium (DFE). Similarly, $\mathcal{R}_0$ for the SEIRS model eqs. \eqref{eq:SEIRS}, with births, deaths, and waning immunity, is 

\begin{equation}\label{eq:R0.SEIRS}
\mathcal{R}_0 = \frac{\beta\,S_0\,r_E}{(r_E+\mu)(r_I+\mu)} = \!\!\!\!\underbrace{\beta\,S_0}_{\substack{\text{rate of new}\\\substack{\text{infections}}}}\!\!\!\!\!\!\!\!\!\!\!\!\overbrace{ \frac{r_E}{r_E+\mu} }^{\substack{\text{latent period}\\\substack{\text{survival probability}}}}\!\!\!\!\!\!\!\!\!\!\!\!\!\!\underbrace{\frac{1}{r_I+\mu},}_{\substack{\text{mean infectious}\\\substack{\text{period}}}}
\end{equation}

which can be similarly factored and interpreted but includes an additional term for the expected fraction of individuals who survive the latent period to become infectious.

To find $\mathcal{R}_0$ for eqs. \eqref{eq:SEIRSlct}, the SEIRS model with Erlang distributed latent and infectious periods, we again use the next generation operator approach \citet{vdDW2002}, but we assume fixed values of $k_E$ and $k_I$. This yields $\mathcal{R}_0$ expressions for each particular case considered, which each have an \textit{ad hoc} interpretation of terms similar to eq. \eqref{eq:R0.SEIRS}. While it is possible to inspect a few specific cases and conjecture a general expression for $\mathcal{R}_0$ in such situations, it is often the case that in practice no general expression for $\mathcal{R}_0$ is formally obtained \citep[e.g.,][]{Johnson2016}. In section \ref{sec:results} below, we show how $\mathcal{R}_0$ can be formally obtained by writing the models in a form suggested by the GLCT (Theorem \ref{th:glct}).

Consider the $\mathcal{R}_0$ expression obtained from eqs. \eqref{eq:SEIRSlct} in the specific case where $k_E=3$, $k_I=4$, and rates $r_{E,k_E}=k_E/\tau_E$ and $r_{I,k_{I}}=k_I/\tau_I$. In this case,  
\begin{equation}
\mathcal{R}_0=\beta\,S_0\,\frac{(r_{E,3})^3}{(r_{E,3}+\mu)^3}\bigg(\frac{1}{r_{I,4}+\mu}+\frac{r_{I,4}}{(r_{I,4}+\mu)^2}+\frac{{r_{I,4}}^2}{(r_{I,4}+\mu)^3}+\frac{{r_{I,4}}^3}{(r_{I,4}+\mu)^4}\bigg).
\end{equation}
This might prompt one to conjecture that, for arbitrary positive integers $k_E$ and $k_I$, 
\begin{equation} \label{eq:SEIRSlctR0}
\mathcal{R}_0=\beta\,S_0\,\frac{(r_{E,k_E})^{k_E}}{(r_{E,k_E}+\mu)^{k_E}}\bigg(\frac{1}{r_{I,k_I}+\mu}+\frac{r_{I,k_I}}{(r_{I,k_I}+\mu)^2}+\dots+\frac{{r_{I,k_I}}^{k_I-1}}{(r_{I,k_I}+\mu)^{k_I}}\bigg).
\end{equation}
However this is only a conjecture without any additional supporting analyses. Such terms can also sometimes be challenging to interpret with confidence. As with the $\mathcal{R}_0$ expressions for the two simpler models above, we can interpret this expression \textit{ad hoc} as the per-infectious-individual rate of new infections, times the probability of surviving all sub-states of E, times the expected period of infectiousness. However, in practice, the proper interpretation may not be obvious.  

In section \ref{sec:SEIRSpt}, we introduce a novel, GLCT-based phase-type distributed SEIRS model (a further generalization of eqs. \eqref{eq:SEIR} - \eqref{eq:SEIRSlct}), and illustrate how the natural matrix-vector formulation of that model can be used to find and interpret a very general expression for $\mathcal{R}_0$. In doing so, as detailed in section \ref{sec:SEIRSconjproof}, we confirm that the above conjecture and $\mathcal{R}_0$ interpretation regarding eqs. \eqref{eq:SEIRSlctR0} holds true for arbitrary positive integers $k_E$ and $k_I$. But first, let us turn our attention to the derivation of a similarly general (and in some ways, simpler) expression for the predator reproduction number $\mathcal{R}_{pred}$ for our generalized predator-prey model.

\subsection{Predator-prey models}

In \citet{HurtadoRichards2021}, an extension of the Rosenzweig-MacArthur model was introduced which assumes an immature/mature stage structure in the predator population. Here we consider a slight extension of that model (the aforementioned model is the $\eta=0$ case of the equations below) given by
\begin{subequations} \label{eq:RMPT} \begin{align} 
	\frac{dN}{dt} =&\; r\,N\bigg(1-\frac{N}{K}\bigg)-\frac{a\,(P_m+\eta\,P_{imm})}{h+N}N\\
	\frac{d\mathbf{x}}{dt} =&\; \chi\frac{a\,N}{h+N}P_m\;\boldsymbol{\alpha}_\mathbf{x} + \mathbf{A_x}^\text{T} \mathbf{x}  \\
	\frac{d\mathbf{y}}{dt} =&\; \underbrace{-\mathbf{1}^\text{T}\mathbf{A_x}^\text{T}\mathbf{x}}_\text{scalar}\;\boldsymbol{\alpha}_\mathbf{y} + \mathbf{A_y}^\text{T}\mathbf{y}  .\label{eq:RMPTc}
	\end{align} \end{subequations} 

The prey population ($N$) follows a logistic growth model in the absence of immature predators ($P_{imm}$) and mature predators ($P_m$), and are removed by predators according to a Holling type-II functional response, with maximum predation rate $a$ and half saturation constant $h$. For simplicity, immature predators have no mortality, are tracked by the $k$ sub-state variables in vector $\mathbf{x}=[x_1,\ldots,x_k]^\text{T}$ ($P_{imm}=\sum x_i$), and have a maturation time that is phase-type distributed with parameter vector $\boldsymbol{\alpha_\mathbf{x}}$ and matrix $\mathbf{A_x}$. Recall that, if $\alpha_{\mathbf{x}*}$ is the sum of the elements of $\boldsymbol{\alpha}_\mathbf{x}$, then $\alpha_{\mathbf{x}*}>0$ implies only a fraction $1-\alpha_{\mathbf{x}*}$ of the predators born enter the immature predator sub-states. One could assume the remainder die before reaching this stage, or one could add an additional term to eq. \eqref{eq:RMPTc} if that fraction is able to skip the immature stages and directly enter the mature predator sub-states. For simplicity, below we will assume $\alpha_{\mathbf{x}*}=0$. Once mature, predators then survive as adults for a period of time that is also phase-type distributed (tracked by $m$ state variables in the vector $\mathbf{y}=[y_1,\ldots,y_m]^\text{T}$, where $P_{m}=\sum y_i$) defined by parameter vector $\boldsymbol{\alpha_\mathbf{y}}$ and matrix $\mathbf{A_y}$ \citep{HurtadoRichards2021}. 

In the next section, we show how reproduction numbers for this generalized predator-prey model can be obtained for arbitrary phase-type distributions, and then we will revisit the similarly generalized SEIRS model discussed above.

\section{Results}\label{sec:results}

\subsection{Predator-prey model with phase-type predator maturation time and lifetime}

Mathematically, the next generation operator approach for finding basic reproduction numbers in \citet{vdDW2002} can be used to find the population reproduction numbers in a model like eqs. \eqref{eq:RMPT}. The results of applying this technique to find the predator reproduction number $\mathcal{R}_\text{pred}$ for this model are summarized in the following theorem.

\begin{theorem} \label{Th:RMR} The predator-free (prey-only) equilibrium state for eqs. \eqref{eq:RMPT} is locally asymptotically stable if $\mathcal{R}_\text{pred} < 1$, but unstable if $\mathcal{R}_\text{pred}>1$, where $\mathcal{R}_\text{pred}$ is given by
	
	\begin{equation}\label{eq:Rpred} \mathcal{R}_\text{pred}=\underbrace{\overbrace{\chi\left(1-\alpha_{\mathbf{x}*}\right)\,\frac{aN}{h+N}}^\text{birth rate} \; \overbrace{\left({\boldsymbol{\alpha_\mathbf{y}}}^\text{T}\, (-\mathbf{A_\mathbf{y}})^{-1}\, \boldsymbol{1}\right)}^{\substack{\text{mean time predators}\\\substack{\text{spend as adults}}}}}_\text{mean number of offspring per predator}. \end{equation}
	
	The product of the first two terms in eq. \eqref{eq:Rpred} is the expected number of new immature predators created by a single mature predator over the average predator's reproductive lifetime, when introduced to a prey population at the predator-free (prey-only) equilibrium. The third term is the fraction of predators that survive to reach maturity. Here $1-\alpha_{\mathbf{x}*}$ is the sum of entries in $\boldsymbol{\alpha_\mathbf{x}}$ (see eqs. \eqref{eq:pt}).
	
\end{theorem}

\begin{proof} 
Using the next generation operator approach detailed in \citet{vdDW2002} to find $\mathcal{R}_\text{pred}$, we first rewrite the system as 
\begin{displaymath}
	\begin{pmatrix} \dot{\mathbf{x}} \\ \dot{\mathbf{y}} \\\dot{\mathbf{N}} \end{pmatrix}=\underbrace{\begin{pmatrix} \boldsymbol{\alpha}_\mathbf{x}\;\chi\frac{a\,N}{h+N}P_m \\ 0 \\ 0 \end{pmatrix}}_{\mathcal{F}} -\left[ 
	\underbrace{\begin{pmatrix} -\mathbf{D}_{\text{diag}(\mathbf{A_x})} \mathbf{x} \\ -\mathbf{D}_{\text{diag}(\mathbf{A_\mathbf{y}})} \mathbf{y} \\ \frac{a\,N}{h+N}(P_m+\eta\,P_{imm}) \end{pmatrix}}_{\mathcal{V}^{-}} - 
	\underbrace{\begin{pmatrix} {(\mathbf{A_x}-\mathbf{D}_{\text{diag}(\mathbf{A_x})})}^\text{T} \mathbf{x} \\ \boldsymbol{\alpha_\mathbf{y}} (-\mathbf{1}^\text{T}\mathbf{A_{x}}^\text{T} \mathbf{x}) + (\mathbf{A_\mathbf{y}}-\mathbf{D}_{\text{diag}(\mathbf{A_\mathbf{y}})})^\text{T} \mathbf{y} \\ r\,N\bigg(1-\frac{N}{K}\bigg) \end{pmatrix}}_{\mathcal{V}^{+}} \right]
\end{displaymath} where $\mathcal{F}$, $\mathcal{V}^{-}$, and $\mathcal{V}^{+}$ satisfy the requirements of Theorem 2 in \citet{vdDW2002}. 

The Jacobians of $\mathcal{F}$ and $\mathcal{V}=\mathcal{V}^{-}-\mathcal{V}^{+}$ evaluated at the prey-only equilibrium yield the matrices

\begin{equation}
	\mathbf{F}=\begin{bmatrix} \mathbf{0}_{|\mathbf{x}|\times|\mathbf{x}|} & \boldsymbol{\alpha_\mathbf{x}} {\boldsymbol{1}_{|\mathbf{y}|}}^\text{T} \chi \frac{aN}{h+N}  \\ \mathbf{0}_{|\mathbf{y}|\times|\mathbf{x}|} & \mathbf{0}_{|\mathbf{y}|\times|\mathbf{y}|}\end{bmatrix} \quad \text{ and } \quad
	\mathbf{V}=\begin{bmatrix} - \mathbf{A_\mathbf{x}}^\text{T} & \boldsymbol{0} \\ \boldsymbol{\alpha_\mathbf{y}} {\boldsymbol{1}_{|\mathbf{x}|}}^\text{T} \mathbf{A_\mathbf{x}}^\text{T} & - \mathbf{A_\mathbf{y}}^\text{T}\end{bmatrix}.
\end{equation}

The subscripts on the matrices indicates their dimensions (e.g., $|\mathbf{x}|\times|\mathbf{y}|$ indicates there are as many rows as immature predator sub-states, and as many columns as mature predator sub-states). Below we also use the notation $\mathbf{I}$ to indicate an appropriately sized identity matrix, and $\mathbf{1}_k$ to indicate a column vector of ones, that is length $k$.

$\mathcal{R}_{pred}$ is the spectral radius (i.e, the largest eigenvalue modulus) of $\mathbf{FV}^{-1}$ \citep{vdDW2002}. Recalling that the general form for the inverse (when it exists) of a block matrix like eq. \eqref{eq:V} is \begin{equation} \label{inverseblockmatrix} {\begin{bmatrix} \mathbf{A} & 0 \\ \mathbf{C} & \mathbf{D} \end{bmatrix}}^{-1}= \; \begin{bmatrix} \mathbf{A}^{-1} & \mathbf{0} \\ -\mathbf{D}^{-1}\mathbf{C}\,\mathbf{A}^{-1} & \mathbf{D}^{-1} \end{bmatrix} \end{equation} for generic, appropriately sized matrix blocks $\mathbf{A}$, $\mathbf{C}$, $\mathbf{D}$. This gives \begin{equation} 
	\mathbf{V}^{-1}=\begin{bmatrix} (-\mathbf{A_x}^\text{T})^{-1} & \mathbf{0} \\ 
		-(-\mathbf{A_y}^\text{T})^{-1} \boldsymbol{\alpha_\mathbf{y}}\, {\mathbf{1}_{|x|}}^\text{T}\mathbf{A_x}^\text{T} (\mathbf{-A_x}^\text{T})^{-1} \quad & (-\mathbf{A_y}^\text{T})^{-1} \end{bmatrix} \end{equation}
which yields that 

\begin{equation} \mathbf{FV}^{-1}=\begin{bmatrix}  \boldsymbol{\alpha_\mathbf{x}}\, {\mathbf{1}_{|\mathbf{y}|}}^\text{T} \chi \, \frac{aN}{h+N}\, (-\mathbf{A_y}^\text{T})^{-1} \boldsymbol{\alpha_\mathbf{y}}\, {\mathbf{1}_{|\mathbf{x}|}}^\text{T} \quad & \boldsymbol{\alpha_\mathbf{x}} {\boldsymbol{1}_{|\mathbf{y}|}}^\text{T} \chi \, \frac{aN}{h+N}\, (-\mathbf{A_y}^\text{T})^{-1} \\ \mathbf{0} & \mathbf{0} \end{bmatrix} \end{equation}

is an upper triangular block matrix with a matrix of zeros as one of its diagonal blocks, so any nonzero eigenvalues will come from the upper left block of the matrix. It follows that  \begin{equation}\mathcal{R}_\text{pred}=\rho \left( \boldsymbol{\alpha_\mathbf{x}} \boldsymbol{1}^\text{T} \chi \, \frac{aN}{h+N}\, (-\mathbf{A_y}^\text{T})^{-1} \boldsymbol{\alpha_\mathbf{y}}\, {\mathbf{1}_{|\mathbf{x}|}}^\text{T}  \right),
\end{equation}

where $\rho(\mathbf{M})$ denotes the spectral radius of matrix $\mathbf{M}$. Note that, according to eqs. \eqref{eq:pt},  \begin{equation} \boldsymbol{1}^\text{T}\, (-\mathbf{A_y}^\text{T})^{-1}\, \boldsymbol{\alpha_\mathbf{y}}={\boldsymbol{\alpha_\mathbf{y}}}^\text{T}\, (-\mathbf{A_y})^{-1}\, \boldsymbol{1} \end{equation} is scalar value and can be interpreted as the mean survival time of predators after they reach maturity. Factoring this quantity and scalar $\chi\,\frac{aN}{h+N}$ out of the spectral radius computation yields 
\begin{equation}
	\mathcal{R}_\text{pred}=\chi\,\frac{aN}{h+N}\;\left({\boldsymbol{\alpha_\mathbf{y}}}^\text{T}\, (-\mathbf{A_y})^{-1}\, \boldsymbol{1}\right)\;\left[\rho \left( \boldsymbol{\alpha_\mathbf{x}}\, {\mathbf{1}_{|\mathbf{x}|}}^\text{T} \right)\right].
\end{equation}

Since $\boldsymbol{\alpha_\mathbf{x}}$ is a column vector, and ${\mathbf{1}_{|\mathbf{x}|}}^\text{T}$ is a row vector of the same length, then the spectral radius of square matrix $\boldsymbol{\alpha_\mathbf{x}} {\mathbf{1}_{|\mathbf{x}|}}^\text{T}$ is\footnote{It is known that an $n\times n$ matrix that can be written as a length $n$ column vector times a length $n$ row vector has $n-1$ zero eigenvalues and the remaining eigenvalue is the dot product of those two vectors.} the dot product of $\boldsymbol{\alpha_\mathbf{x}}$ and ${\mathbf{1}_{|\mathbf{x}|}}$, which equals the sum of the entries in $\boldsymbol{\alpha_\mathbf{x}}$. Recalling the definition of $\alpha_*$ from eqs. \eqref{eq:pt}, define $\alpha_{\mathbf{x}*}=1-\rho(\boldsymbol{\alpha_\mathbf{x}} {\mathbf{1}_{|\mathbf{x}|}}^\text{T})$. Then it follows that  \begin{equation}\mathcal{R}_\text{pred}=\chi\,\frac{aN}{h+N} \left(\boldsymbol{\alpha_\mathbf{y}}\, (-\mathbf{A_y})^{-1}\, \boldsymbol{1}\right) \;\left(1-\alpha_{\mathbf{x}*}\right).\end{equation}
\end{proof}

Note that one would assume that $\alpha_{\mathbf{x}*}=0$. To assume otherwise would require assuming that a fraction of predators perish before entering the immature stage \textendash{} which would be equivalent to normalizing $\boldsymbol{\alpha}$ so that all entries sum to one, and assuming an effective conversion factor of $\chi_*=\chi(1-\alpha_{\mathbf{x}*})$ \textendash{} or adding an appropriate term to the mature predator stage assuming this fraction skips the immature stages altogether and enters directly into the mature states.  Also note that the mean time predators spend as adults incorporates the fraction that survive to reach maturity.

Observe that the above predator reproduction number expression holds for any choice of phase-type distribution assumptions to describe the predator stage structure (which determines the dimension of model eqs. \eqref{eq:RMPT}). Here we have used a somewhat simplistic set of assumptions regarding the predator birth rate and survival of offspring to maturity to more clearly illustrate the process of computing the reproduction number using the next generation operator approach, and how the expression obtained can be interpreted using properties of phase-type distributions and associated Markov chains. 

We next consider the slightly more complex SEIRS model, which further illustrates the implementation of this approach.

\subsection{Generalized SEIRS model with phase-type dwell times in E, I, and R}\label{sec:SEIRSpt}

Consider the SEIRS model eqs. \eqref{eq:SEIRSlct}. We apply the GLCT to generalize this model as follows (cf. the steps used to derive the simpler SEIR model in \citet{HurtadoRichards2021} and the procedure described in \citet{HurtadoRichards2020b}). Assume that the latent period (time spent in state E) follows a phase-type distribution parameterized by a length $k_E$ vector $\boldsymbol{\alpha_E}$ and $k_E\times k_E$ matrix $\mathbf{A_E}$, the infectious period (time spent in state I) follows a phase-type distribution parameterized by a length $k_I$ vector $\boldsymbol{\alpha_I}$ and $k_I\times k_I$ matrix $\mathbf{A_I}$, and the duration of immunity after recovery (time spent in state R) follows a phase-type distribution defined by a length $k_R$ vector $\boldsymbol{\alpha_R}$ and $k_R\times k_R$ matrix $\mathbf{A_R}$.

The resulting mean field ODEs for this generalized SEIRS model are
\begin{subequations} \label{eq:SEIRSpt}\begin{align}
\frac{dS}{dt} =& \; \Lambda(S,\mathbf{x},\mathbf{y},\mathbf{z}) - \mu\,S - \lambda(t)\,S + \underbrace{\big(-\mathbf{1}^\text{T}\mathbf{A_R}^\text{T}\mathbf{z}\big)}_{\text{scalar}}  \\
\frac{d\mathbf{x}}{dt} =& \; \boldsymbol{\alpha_E}\,\lambda(t)\,S + {\mathbf{A_E}}^\text{T}\mathbf{x} - \mu\,\mathbf{x} \\
\frac{d\mathbf{y}}{dt} =& \; \boldsymbol{\alpha_I}\underbrace{\big(-\mathbf{1}^\text{T}\mathbf{A_E}^\text{T}\mathbf{x}\big)}_\text{scalars} + {\mathbf{A_I}}^\text{T}\mathbf{y} - \mu\,\mathbf{y} \\
\frac{d\mathbf{z}}{dt} =& \; \boldsymbol{\alpha_R}\overbrace{\big(-\mathbf{1}^\text{T}\mathbf{A_I}^\text{T}\mathbf{y}\big)} + {\mathbf{A_R}}^\text{T}\mathbf{z} - \mu\,\mathbf{z} .
\end{align}\end{subequations}

Here $\mathbf{x}=[E_1,\ldots,E_{k_E}]^\text{T}$, $\mathbf{y}=[I_1,\ldots,I_{k_I}]^\text{T}$, and $\mathbf{z}=[R_1,\ldots,R_{k_R}]^\text{T}$ are the column vectors of sub-states of E, I, and R, respectively, where $E=\sum E_i$, $I=\sum I_i$, and $R=\sum R_i$. Also, $\boldsymbol\beta=[\beta_1,\ldots,\beta_{k_I}]^\text{T}$ and the force of infection is $\lambda(t) = \beta_1\,I_1(t) + \cdots + \beta_{k_I}\,I_{k_I}(t) =\,\boldsymbol{\beta}\cdot\mathbf{y}=\,\boldsymbol{\beta}^\text{T}\mathbf{y}$.

We have also assumed a general birth rate $\Lambda(S,\mathbf{x},\mathbf{y},\mathbf{z})\geq0$ which we assume is a sufficiently smooth function that also yields a locally asymptotically stable disease free equilibrium (DFE) ($S_0,\mathbf{0},\mathbf{0},\mathbf{0}$) in the absence of the pathogen. The terms labeled as scalars in eqs. \eqref{eq:SEIRSpt} are the sums of terms in each loss rate vector from the different exposed, infectious, and recovered (immune) states (cf. vector $\mathbf{a}$ in eq. \eqref{eq:ratematrix} and see Theorem \ref{th:glct}).

Using the next generation operator approach \citep{vdDW2002} to find $\mathcal{R}_0$ for this model yields the following result, where $\mathbf{G_E} = \mathbf{A_E}-\mu\,\mathbf{I}_{k_E\times k_E}$ and $\mathbf{G_I} = \mathbf{A_I}-\mu\,\mathbf{I}_{k_I\times k_I}$ (we will discuss the interpretation of these matrices below).

\begin{theorem} \label{Th:R0} The DFE for eqs. \eqref{eq:SEIRSpt} is locally asymptotically stable if $\mathcal{R}_0 < 1$, but unstable if $\mathcal{R}_0>1$, where $\mathcal{R}_0$ is given by
\begin{equation}\mathcal{R}_0=\; \mathcal{R}_\text{0,new}  \; \mathcal{P}_\text{E$\to$I}  \end{equation}
where $\mathcal{R}_\text{0,new}$ is the expected number of new infections created by a single infectious individual over the course of the average predator lifetime when introduced to the population at the DFE, given by 
\begin{equation} \mathcal{R}_\text{0,new} = \boldsymbol{\alpha_I}^\text{T}(-\mathbf{G_I}^{-1})(\boldsymbol{\beta}S_0),\end{equation}
and $\mathcal{P}_\text{E$\to$I}$ is the probability of surviving the exposed state (E) and transitioning into the infectious state (I), which is given by the dot product
\begin{equation} \mathcal{P}_\text{E$\to$I} =\; \boldsymbol{\alpha_E}  \cdot ((-\mathbf{G_E}^{-1})(- \mathbf{A_E}\mathbf{1}_{k_E})) .\end{equation}
	
\end{theorem}

\begin{proof} The right hand side of eqs. \eqref{eq:SEIRSpt} can be written as follows, yielding $\mathcal{F}$ and $\mathcal{V}=\mathcal{V}^{-} -\mathcal{V}^{+}$ which satisfy the requirements of Theorem 2 in \citet{vdDW2002}.
	
\begin{displaymath}
\begin{pmatrix} \dot{\mathbf{x}} \\ \dot{\mathbf{y}} \\\dot{\mathbf{z}} \\ \dot{S} \end{pmatrix}=\underbrace{\begin{pmatrix} \vec{\alpha_E} \lambda(t) S \\ 0 \\ 0 \\ 0 \end{pmatrix}}_{\mathcal{F}} -\left[ 
        \underbrace{\begin{pmatrix} -\mathbf{D}_{\text{diag}(\mathbf{A_E})} \mathbf{x} +\mu \, \mathbf{x} \\ -\mathbf{D}_{\text{diag}(\mathbf{A_I})} \mathbf{y} +\mu \, \mathbf{y}  \\ -\mathbf{D}_{\text{diag}(\mathbf{A_R})} \mathbf{z} +\mu \, \mathbf{z}  \\ \lambda(t) S + \mu S \end{pmatrix}}_{\mathcal{V}^{-}} - 
        \underbrace{\begin{pmatrix} {(\mathbf{A_E}-\mathbf{D}_{\text{diag}(\mathbf{A_E})})}^\text{T} \mathbf{x} \\ \boldsymbol{\alpha_I} (-\mathbf{1}^\text{T}\mathbf{A_{E}}^\text{T} \mathbf{x}) + (\mathbf{A_I}-\mathbf{D}_{\text{diag}(\mathbf{A_I})})^\text{T} \mathbf{y} \\ \boldsymbol{\alpha_R} (-\mathbf{1}^\text{T}\mathbf{A_{I}}^\text{T} \mathbf{y}) + (\mathbf{A_R}-\mathbf{D}_{\text{diag}(\mathbf{A_R})})^\text{T} \mathbf{z} \\ \Lambda(S,\mathbf{x},\mathbf{y},\mathbf{z}) + -\mathbf{1}^\text{T}\mathbf{A_R}^\text{T}\mathbf{z} \end{pmatrix}}_{\mathcal{V}^{+}} \right],
\end{displaymath}
The upper left block ($\mathbf{F}$) of the Jacobian for $\mathcal{F}$ evaluated at the DFE is 
\begin{equation}
\mathbf{F}=\begin{bmatrix} \mathbf{0}_{k_E\times k_E} & \boldsymbol{\alpha_E} \boldsymbol{\beta}^{T}S_0  \\ \mathbf{0}_{k_I\times k_E} & \mathbf{0}_{k_I\times k_I}\end{bmatrix}= 
  \begin{bmatrix} \mathbf{0} & \boldsymbol{\alpha_E}\boldsymbol{\beta}^{T} \\ \mathbf{0} & \mathbf{0} \end{bmatrix} S_0\end{equation} and the upper left block ($\mathbf{V}$) of the Jacobian for $\mathcal{V}$ evaluated at the DFE is \begin{equation} \label{eq:V}
\mathbf{V}=\begin{bmatrix} \mu \mathbf{I} -\mathbf{A_E}^\text{T} & \mathbf{0} \\ \boldsymbol{\alpha_I} \mathbf{1}_{k_E}^\text{T}\mathbf{A_{E}}^\text{T} &  \mu \mathbf{I} -\mathbf{A_I}^\text{T}\end{bmatrix}=
           \begin{bmatrix}  -\mathbf{G_E}^\text{T} & \mathbf{0} \\ \boldsymbol{\alpha_I} \mathbf{1}_{k_E}^\text{T}\mathbf{A_{E}}^\text{T} &  -\mathbf{G_I}^\text{T}\end{bmatrix}. \end{equation}
Using eq. \eqref{inverseblockmatrix} yields
\begin{equation} 
\mathbf{V}^{-1}=\begin{bmatrix} (-\mathbf{G_E}^{-1})^\text{T} & \mathbf{0} \\ 
(-\mathbf{G_I}^{-1})^\text{T} (-\boldsymbol{\alpha_I} \mathbf{1}_{k_E}^\text{T}\mathbf{A_E}^\text{T}) (-\mathbf{G_E}^{-1})^\text{T} & (-\mathbf{G_I}^{-1})^\text{T}  \end{bmatrix}. \end{equation}
The basic reproduction number is the spectral radius of $\mathbf{FV}^{-1}$, i.e., $\mathcal{R}_0=\rho(\mathbf{FV}^{-1})$, where
\begin{equation} \mathbf{FV}^{-1} = \begin{bmatrix} \boldsymbol{\alpha_E} \boldsymbol{\beta}^{T}S_0 \; (-\mathbf{G_I}^{-1})^\text{T} (-\boldsymbol{\alpha_I} \mathbf{1}_{k_E}^\text{T}\mathbf{A_E}^\text{T}) (-\mathbf{G_E}^{-1})^\text{T} & \boldsymbol{\alpha_E} \boldsymbol{\beta}^{T}S_0 \; (-\mathbf{G_I}^{-1})^\text{T}  \\ \mathbf{0} & \mathbf{0} \end{bmatrix}.\end{equation}
Observe that, as in the previous example, $\mathbf{FV}^{-1}$ is an upper triangular block matrix with a diagonal that is all zeros except for the top left block. Therefore,  \begin{equation}\label{eq:SEIRSptR01}\begin{aligned}\mathcal{R}_0=&\;\rho \big( \boldsymbol{\alpha_E} \boldsymbol{\beta}^{T}S_0 \; (-\mathbf{G_I}^{-1})^\text{T} (-\boldsymbol{\alpha_I} \mathbf{1}_{k_E}^\text{T}\mathbf{A_E}^\text{T}) (-\mathbf{G_E}^{-1})^\text{T} \big) \\
=&\;\rho \big( \boldsymbol{\alpha_E} \big(\boldsymbol{\alpha_I}^\text{T}(-\mathbf{G_I}^{-1})(\boldsymbol{\beta}S_0)\big)^\text{T} \;  (- \mathbf{A_E}\mathbf{1}_{k_E})^\text{T} (-\mathbf{G_E}^{-1})^\text{T} \big) \end{aligned}\end{equation}

The above expression for $\mathcal{R}_0$ can be simplified as follows. 

First, recall the expected reward equation, eq. \eqref{eq:Ereward}. The term in eq. \eqref{eq:SEIRSptR01}, \begin{equation}\label{eq:Ireward} \boldsymbol{\alpha_I}^\text{T}(-\mathbf{G_I}^{-1})(\boldsymbol{\beta}S_0)\end{equation} is the expected reward accrued prior to reaching an absorbing state, for a reward Markov process associated with a phase-type distribution parameterized by vector $\boldsymbol{\alpha_I}$ and matrix $\mathbf{G_I} = \mathbf{A_I}-\mu\,\mathbf{I}_{k_I\times k_I}$ with reward rate vector $\boldsymbol\beta\,S_0$. This phase-type distribution describes the duration of time spent in the infected state, and is formally the phase-type distribution that is the minimum (see eq. \eqref{eq:minPT}) of an exponential distribution with rate $\mu$ and a phase-type distribution with parameters $\boldsymbol{\alpha_I}$ and $\mathbf{A_I}$. Therefore, eq. \eqref{eq:Ireward} is \textit{the expected number of new infections created by an average infectious individual over the duration of the mean infectious period}. This quantity is a $1\times1$ matrix, and can be treated as a scalar and factored out of eq. \eqref{eq:SEIRSptR01}. 

If we denote this expected number of new infections, eq. \eqref{eq:Ireward}, as \begin{equation}\label{eq:SEIRSpptR02} \mathcal{R}_\text{0,new} = \boldsymbol{\alpha_I}^\text{T}(-\mathbf{G_I}^{-1})(\boldsymbol{\beta}S_0) = \boldsymbol{\alpha_I}^\text{T}(-(\mathbf{A_I}-\mu\,\mathbf{I}_{k_I\times k_I})^{-1})(\boldsymbol{\beta}S_0),\end{equation} then it follows that eq. \eqref{eq:SEIRSptR01} can be rewritten
\begin{equation}\label{eq:SEIRSpptR03}
\mathcal{R}_0=\; \mathcal{R}_\text{0,new} \; \rho \bigg( \boldsymbol{\alpha_E} (- \mathbf{A_E}\mathbf{1}_{k_E})^\text{T} (-\mathbf{G_E}^{-1})^\text{T} \bigg). \end{equation}

Let $\mathbf{Q_E}$ denote the transient block of the transition probability matrix for the embedded jump process associated with the phase-type distribution parameterized by $\boldsymbol{\alpha_E}$ and $\mathbf{G_E}$. Denote the vector of dwell time rates for the E sub-states by $\mathbf{v_E}=-\text{diag}(\mathbf{G_E})$, and let $\mathbf{D_{v_E}}$ be the diagonal matrix with diagonal $\mathbf{v_E}$. Then eq. \eqref{eq:SEIRSpptR03} can be rewritten using eq. \eqref{eq:AQ} as 
\begin{equation}\label{eq:SEIRSpptR04}\begin{split}
\mathcal{R}_0=&\; \mathcal{R}_\text{0,new} \; \rho \bigg( \boldsymbol{\alpha_E} (- \mathbf{A_E}\mathbf{1}_{k_E})^\text{T} ((\mathbf{D}_\mathbf{v_E}(\mathbf{I}-\mathbf{Q_E}))^{-1})^\text{T} \bigg) \\
=&\;\mathcal{R}_\text{0,new} \; \rho \bigg( \boldsymbol{\alpha_E} ((\mathbf{I}-\mathbf{Q_E})^{-1} {\mathbf{D}_\mathbf{v_E}}^{-1}(- \mathbf{A_E}\mathbf{1}_{k_E}))^\text{T} \bigg) \\
\end{split}\end{equation}

Let $r_i$ denote the loss rate from the $i^\text{th}$ sub-state of E excluding deaths (i.e., assuming $\mu =0$), i.e., it is the $i^\text{th}$ diagonal entry of $\mathbf{A}_E$. Then the $i^\text{th}$ entry of $\mathbf{v_E}$ is $r_i+\mu$, and so by eq. \eqref{eq:q} the $i^\text{th}$ entry in column vector \begin{equation}\label{eq:blarg}
{\mathbf{D}_\mathbf{v_E}}^{-1}(- \mathbf{A_E}\mathbf{1}_{k_E})
\end{equation} is the probability of leaving state E from the $i^\text{th}$ sub-state of E, times the probability of then entering state I (given by $r_i/(r_i+\mu)$) as opposed to entering the deceased state (not tracked in the model). Thus, eq. \eqref{eq:blarg} is the column of matrix $\mathbf{R}$ (recall eq. \eqref{eq:absprob}) corresponding to the \textit{infected} (as opposed to \textit{deceased}) absorbing state. 

Thus, according to eq. \eqref{eq:absprob}, the column vector in which the $i^\text{th}$ entry is the probability of surviving the exposed state (E) and entering the infectious state (I), given that the initial state was the $i^\text{th}$ transient state, is the vector \begin{equation}\label{eq:UE} \mathbf{U_E} =\; ((\mathbf{I}-\mathbf{Q_E})^{-1} {\mathbf{D}_\mathbf{v_E}}^{-1}(- \mathbf{A_E}\mathbf{1}_{k_E}). \end{equation}

Therefore, it follows that eq. \eqref{eq:SEIRSpptR04} can be written in terms of these survival probabilities $\mathbf{U_E}$ as 
\begin{equation}
\mathcal{R}_0=\;\mathcal{R}_\text{0,new} \; \rho \bigg( \boldsymbol{\alpha_E} (\mathbf{U_E})^\text{T} \bigg)  = \; \mathcal{R}_\text{0,new} \; \boldsymbol{\alpha_E} \cdot \mathbf{U_E}   
\end{equation}
where the spectral radius of the scalar product  $\rho \big( \boldsymbol{\alpha_E} (\mathbf{U_E})^\text{T} \big)=\boldsymbol{\alpha_E} \cdot \mathbf{U_E}$ is the sum of each initial state probability, $\alpha_i$, times the exposed-state survival probability, $U_{\mathbf{E}i}$. Thus,  \begin{equation}
\mathcal{P}_{E \to I}=\boldsymbol{\alpha_E} \cdot \mathbf{U_E}
\end{equation} is the overall probability of surviving the exposed state E and transitioning to the infectious state I, and therefore

\begin{equation}\begin{split}
\mathcal{R}_0=\;& \underbrace{\big(\boldsymbol{\alpha_I}^\text{T}(-\mathbf{G_I}^{-1})(\boldsymbol{\beta}S_0)\big)}_{\mathcal{R}_{0,new}}  \underbrace{\boldsymbol{\alpha_E}  \cdot ((-\mathbf{G_E}^{-1})(- \mathbf{A_E}\mathbf{1}_{k_E}))}_{\mathcal{P}_{E \to I}} \\
 =\;& \underbrace{\big(\boldsymbol{\alpha_I}^\text{T}(-(\mathbf{A_I}-\mu\,\mathbf{I}_{k_I\times k_I})^{-1})(\boldsymbol{\beta}S_0)\big)}_{\mathcal{R}_{0,new}}  \underbrace{\boldsymbol{\alpha_E}  \cdot ((-(\mathbf{A_E}-\mu\,\mathbf{I}_{k_E\times k_E})^{-1})(- \mathbf{A_E}\mathbf{1}_{k_E}))}_{\mathcal{P}_{E \to I}}.  
\end{split}
\end{equation}

\end{proof}

\subsection{Example: $\mathcal{R}_0$ for the SEIRS model with Erlang latent and infectious periods}\label{sec:SEIRSconjproof}

Observe that the SEIRS model with Erlang latent and infectious periods, and exponentially distributed duration of immunity, eqs. \eqref{eq:SEIRSlct}, is a special case of eqs. \eqref{eq:SEIRSpt}.  Therefore, the following corollary to Theorem \ref{Th:R0} proves the previous conjecture about eqs. \eqref{eq:SEIRSlct} and eq. \eqref{eq:SEIRSlctR0}.

\begin{corollary} The basic reproduction number $\mathcal{R}_0$ for eqs. \eqref{eq:SEIRSlct} \textendash{} the SEIRS model with mortality rate $\mu$, exponentially distributed duration of immunity, and Erlang latent and infectious periods with respective means $\tau_E$ and $\tau_I$ and arbitrary integer-valued shape parameters $k_E,k_I\geq 1$ \textendash{} is given by 

\begin{equation}  \mathcal{R}_\text{0} =  \beta\,S_0\;\bigg(\frac{1}{r_I + \mu} + \frac{r_I}{(r_I + \mu)^2} + \frac{r_I^2}{(r_I + \mu)^3} + \cdots + \frac{r_I^{k_I-1}}{(r_I + \mu)^{k_I}}\bigg)\,\bigg(\frac{r_E}{r_E + \mu}\bigg)^{k_E}.  \end{equation}	

where $r_E=k_E/\tau_E$ and $r_I=k_I/\tau_I$.  The DFE is locally asymptotically stable if $\mathcal{R}_0<1$ and unstable if $\mathcal{R}_0>1$.

The first term, $\beta\,S_0$, is the expected rate of new infections per unit time per infected individual added to a population at the DFE. The second term is the probability of surviving the exposed class (E) and entering the infected class (I). The third term is the expected value of distribution of time sindividuals spend in the infected class \textendash{} which is a Coxian distribution defined as the minimum of an exponential distribution with rate $\mu$ and the Erlang infectious period distribution.
\end{corollary}

\begin{proof} Eqs. \eqref{eq:SEIRSlct} are a particular case of eqs. \eqref{eq:SEIRSpt} with constant birthrate $\Lambda(S,\mathbf{x},\mathbf{y},R)=\Lambda$ and force of infection $\lambda(t) = \beta\,\sum_{i=0}^{k_I} I_i(t)$. As in \citet{HurtadoRichards2021}, the Erlang distributed latent period is phase-type distributed with parameters (cf. eq. \eqref{eq:alphaA})

\begin{equation}\boldsymbol{\alpha_E} = \begin{bmatrix} 1 \\ 0 \\ \vdots \\ 0  \end{bmatrix}  \qquad \text{and} \qquad \mathbf{A_E} = \begin{bmatrix} -r_E & r_E & 0 & \cdots & 0 \\
0 & -r_E & r_E  & \ddots & 0 \\
\vdots & \ddots & \ddots  & \ddots & \ddots \\
0 & 0 & \ddots &  -r_E & r_E \\
0 & 0 & \cdots &  0 & -r_E
\end{bmatrix}. \end{equation}

The Erlang distributed infectious period can be similarly parameterized using 

\begin{equation}\boldsymbol{\alpha_I} = \begin{bmatrix} 1 \\ 0 \\ \vdots \\ 0  \end{bmatrix}  \qquad \text{and} \qquad \mathbf{A_I} = \begin{bmatrix} -r_I & r_I & 0 & \cdots & 0 \\
0 & -r_I & r_I  & \ddots & 0 \\
\vdots & \ddots & \ddots  & \ddots & \ddots \\
0 & 0 & \ddots &  -r_I & r_I \\
0 & 0 & \cdots &  0 & -r_I
\end{bmatrix} \end{equation}

For the exponentially distributed duration of immunity, $\boldsymbol{\alpha_R}=[1]$ and $\mathbf{A_R}=[-\epsilon]$. 

Let $\mathbf{G_I}$ denote the matrix given by $\mathbf{A_I}$ with additional $-\mu$ terms along the diagonal ($\mathbf{G_I}=\mathbf{A_I} - \mu\mathbf{I}$) and similarly $\mathbf{G_E}=\mathbf{A_E} - \mu\mathbf{I}$. Then, by Theorem \ref{Th:minExpErlang} in Appendix \ref{A:minExpErlang}, the parameter pairs $\boldsymbol{\alpha_E}$, $\mathbf{G_E}$ and $\boldsymbol{\alpha_I}$, $\mathbf{G_I}$ define two Coxian distributions describing the effective latent period and infectious period distributions after accounting for losses from deaths. Observe that since all $\beta_i=\beta$ then the vector $\boldsymbol{\beta}=\beta\mathbf{1}$. Then by Theorem \ref{Th:R0} it follows that $\mathcal{R}_0$ is given by the product of 

\begin{equation}\begin{split} \mathcal{R}_\text{0,new} =&\; \boldsymbol{\alpha_I}^\text{T}(-\mathbf{G_I}^{-1})(\boldsymbol{\beta}S_0) \\
                                                       =&\; \beta S_0\,\boldsymbol{\alpha_I}^\text{T}(-\mathbf{G_I}^{-1})\mathbf{1}\end{split}
                                                       \end{equation}
and

\begin{equation}\begin{split} \mathcal{P}_\text{E$\to$I} =&\; \boldsymbol{\alpha_E}  \cdot ((-\mathbf{G_E}^{-1})(- \mathbf{A_E}\mathbf{1}_{k_E})) \\
                       =&\; \boldsymbol{\alpha_E}  \cdot ((-\mathbf{G_E}^{-1})([0 \cdots 0\; 1]^\text{T}r_E)). \end{split}\end{equation}

If we let $Y_I$ denote the time spent in state I, then by either eq. \eqref{eq:pt} or eq. \eqref{eq:Ereward}, 

\begin{equation}
\text{E}(Y_I) =\; \boldsymbol{\alpha_I}^\text{T}(-\mathbf{G_I}^{-1})\mathbf{1} =\;  \bigg(\frac{1}{r_I + \mu} + \frac{r_I}{(r_I + \mu)^2} + \frac{r_I^2}{(r_I + \mu)^3} + \cdots + \frac{r_I^{k_I-1}}{(r_I + \mu)^{k_I}}\bigg).
\end{equation}

Since $-\mathbf{G_E}^{-1}$ is of the form (see Appendix \ref{A:minExpErlang})

\begin{equation} -\mathbf{G_E}^{-1} = \frac{1}{r_E + \mu}\begin{bmatrix} 1 & \frac{r_E}{r_E + \mu} & \big(\frac{r_E}{r_E + \mu}\big)^2 & \cdots & \big(\frac{r_E}{r_E + \mu}\big)^{k_E-1} \\
0 & 1 & \frac{r_E}{r_E + \mu} & \ddots & \big(\frac{r_E}{r_E + \mu}\big)^{k_E-2} \\
\vdots & \ddots & \ddots  & \ddots & \ddots \\
0 & 0 & \ddots &  1 & \frac{r_E}{r_E + \mu}  \\
0 & 0 & \cdots &  0 & 1
\end{bmatrix} \end{equation}

then 
\begin{equation}
(-\mathbf{G_E}^{-1})\begin{bmatrix} 0 \\ \vdots \\ 0 \\ 1\end{bmatrix} r_E  = \; 
\begin{bmatrix} \big(\frac{r_E}{r_E + \mu}\big)^{k_E} \\
\big(\frac{r_E}{r_E + \mu}\big)^{k_E-1} \\
\vdots \\
\big(\frac{r_E}{r_E + \mu}\big)^2  \\
\frac{r_E}{r_E + \mu}
\end{bmatrix}
\end{equation}

and thus, \begin{equation} \mathcal{P}_\text{E$\to$I} =\; \boldsymbol{\alpha_E} \cdot 
(-\mathbf{G_E}^{-1})\begin{bmatrix} 0 \\ \vdots \\ 0 \\ 1\end{bmatrix} r_E  = \;  \begin{bmatrix} 1\\0 \\ \vdots \\ 0 \end{bmatrix} \cdot  \begin{bmatrix} \big(\frac{r_E}{r_E + \mu}\big)^{k_E} \\
\big(\frac{r_E}{r_E + \mu}\big)^{k_E-1} \\
\vdots \\
\big(\frac{r_E}{r_E + \mu}\big)^2  \\
\frac{r_E}{r_E + \mu}
\end{bmatrix} =\; \bigg(\frac{r_E}{r_E + \mu}\bigg)^{k_E}.
\end{equation}

It follows that \begin{equation}  \mathcal{R}_\text{0} =  \beta\,S_0\; \bigg(\frac{1}{r_I + \mu} + \frac{r_I}{(r_I + \mu)^2} + \frac{r_I^2}{(r_I + \mu)^3} + \cdots + \frac{r_I^{k_I-1}}{(r_I + \mu)^{k_I}}\bigg)\,\bigg(\frac{r_E}{r_E + \mu}\bigg)^{k_E}.  \end{equation}

\end{proof}

\section{Discussion}\label{sec:discussion}

Here we have shown how reproduction numbers for ODE models of arbitrary finite dimension can be derived and interpreted using the generalized linear chain trick \citep[GLCT;][]{Hurtado2019,HurtadoRichards2020b} in conjunction with continuous time Markov chain (CTMC) theory \citep{Resnick2002,Bailey1990,Bladt2017,Reinecke2012a,Taylor1998}, the associated theory of phase-type (i.e., CTMC absorption time) distributions \citep{Bladt2017ch3, Reinecke2012a}, and the next generation operator method \citep{vdDW2002,Diekmann1990,Diekmann2009}. This approach can yield general expressions for basic reproduction numbers ($\mathcal{R}_0$) in families of epidemic models of arbitrary finite dimension, like the SEIRS model with Erlang latent and infectious periods. We also showed that these techniques can be used to find similar threshold quantities (population reproduction numbers) in single- or multi-species ecological models.  

The success of this approach relies on two related features of mean field ODE models that were derived (or could be interpreted) using the GLCT. The first is that formulating the model (or family of models) in a matrix-vector form consistent with the GLCT yields equations that are agnostic of the actual number of model state variables and their corresponding ODE equations (see also \citet{Hyman2005,Bonzi2010}). That form can then be exploited during the application of the next generation operator approach for computing reproduction numbers. The second feature of such models is that the resulting general reproduction number expressions can be simplified and interpreted through the lens of phase-type distributions and Markov chain theory, allowing for the kinds of insights and interpretations that can be important in applications. 

We have also introduced a novel family of predator-prey models with stage structure in the predator population, and an SEIRS family of contagion models with phase-type distributed latent periods, infectious periods, and duration of immunity. 

This general SEIRS model framework could be further generalized (e.g., to consider other assumptions about the functional form of the force of infection, or less simplistic birth and death processes), and some such extensions are already being considered (Hurtado \& Kiefer, \textit{in prep}), however, there are many existing SEIR-type models (or their stochastic analogs) used in applications that are (or are approximately) special cases of this model. These include models currently being applied to the ongoing SARS-CoV-2 pandemic, for example \citet{Bertozzi2020,Renardy2020,Nande2021,Drake2021,Giordano2020} and additional models mentioned in works such as \citet{Cramer2021}. To illustrate the utility of this approach for the mathematical analysis of more traditional (LCT-based) models, and not just for model derivations, we obtained a general $\mathcal{R}_0$ expression for this general family of SEIRS models of arbitrarily many (finite) dimensions.

In ecological applications, it is relatively uncommon to see the next generation operator approach used to compute population reproduction numbers in, e.g., multispecies population models \citep[although, see ][]{Duffy2016}. This may be a reflection of the somewhat low-dimensional nature of these models (or the single-species components of such models), or perhaps it is a consequence of these quantities traditionally being derived during a more direct equilibrium stability analysis of different ecologically relevant steady states. However, when models are extended to include individual heterogeneity in the form of different discrete individual types or sub-populations, or there is a discrete stage structure added to these populations \textit{a la} the LCT or GLCT, the approach outlined above may prove to be a useful addition to the ecological modeler's toolkit. 

It is also worth noting that there are other threshold quantities, such as type reproduction numbers \citep{Heesterbeek2015,Heesterbeek2007,Roberts2003,Shuai2012}, and methods for finding these quantities might also be able to exploit the matrix-vector structure of GLCT-derived models in ways that yield general results comparable to the analyses and results presented here. Furthermore, $\mathcal{R}_0$ can be obtained by other means, which sometimes yield different $\mathcal{R}_0$ expressions (but with the same threshold behavior near 1), especially for age structured models with differential sensitivity and multi-host (e.g., vector-born) diseases \citep{Hyman2000,Hyman2005,Yang2014,Heffernan2005}. Nonetheless, analyses that employ these other methods to find reproduction numbers may still benefit from using the matrix-vector form of models consistent with the GLCT, as well as the application of CTMC and phase-type distribution theory in the interpretation of such results.

Not only can the phase-type form of such models be an useful analytically, but it can be used to speed up computations in some cases, as shown in \citet{HurtadoRichards2021}. Since the matrix and vector parameters of a phase-type distribution can be estimated from data \citep{BuTools2,BuToolswww}, the GLCT can also be used to build approximate empirical dwell time distributions into ODE models (e.g., see \citet{Kim2019} for an example where Coxian distributed latent periods are derived from data and used to formulate a similar SEIR model to the more general one introduced above). Sometimes, it may be that other modeling approaches are more appropriate, e.g., using a distributed delay model based on integral or integro-differential equations can sometimes be preferable \citep[e.g., see][]{Cassidy2019,Cassidy2020arxiv2SIAM,MacDonald1978,MacDonald1989,Hu2018,Cushing1977,Rost2008}. We hope that, in some such cases, these GLCT-based ODE models may still serve as useful approximations or special cases of these more general models.

Mean field ODE models are used throughout the sciences, and those equations reflect much of the structure associated with their corresponding (often unwritten) individual-based, stochastic model analogs. The GLCT provides one way to more clearly see these deeper connections between mean field ODEs and their stochastic counterparts. Here we have attempted to highlight how the GLCT can be a useful tool for developing and analyzing mean field ODE models by leveraging existing theory, statistical tools, and intuition from CTMCs, phase-type distributions, and related stochastic processes. We hope others find these techniques to be useful for the development, analysis, and application of mean field dynamical systems models. 

\textbf{Acknowledgments:} The authors thank Dr. Deena Schmidt, the Mathematical Biology lab group at UNR, Dr. Marisa Eisenberg, and Dr. Michael Cortez for conversations and comments that improved this manuscript.

\textbf{Disclosure statement:} The authors declare that they have no conflict of interest.

\textbf{Funding:} This work was supported by a grant awarded to PJH by the Sloan Scholars Mentoring Network of the Social Science Research Council with funds provided by the Alfred P. Sloan Foundation; and this material is based upon work supported by the National Science Foundation under Grant No. DEB-1929522.

\begin{appendices} 
\numberwithin{equation}{section}
\renewcommand{\theequation}{\Alph{section}\arabic{equation}}\setcounter{equation}{0}

\section{Expected reward for phase-type distributions} \label{A:reward}

The following is a known result, but is included here for completeness and for the benefit of readers unfamiliar with phase-type distributions. 

\begin{theorem} Suppose $g_i$ is the \textit{reward rate} for time spent in the $i^\text{th}$ state ($T_i$) such that the cumulative reward is $g_i\,T_i$. Let $W$ be the reward accrued up to hitting the absorbing state in a CTMC with $n$ transient states and a single absorbing state, with corresponding phase-type distribution given by $\boldsymbol{\alpha}$ and $\mathbf{A}$.  Then $$\mathbf{E}(W) =\; \boldsymbol\alpha^\text{T} (\mathbf{-A^{-1}}) \mathbf{g}.$$
\end{theorem} \begin{proof} Following Theorem 3.1.16 in \citet{Bladt2017}, 
\begin{align*}
\mathbf{E}(W)=&\;\sum_{i=1}^n \alpha_i \mathbf{E}\bigg(\sum_{j=1}^n T_j\,g_j\bigg|\text{initial state = $i^\text{th}$ state}\bigg) \\
             =&\; \sum_{i=1}^n \alpha_i \sum_{j=1}^n u_{ij}g_j = \boldsymbol\alpha^\text{T} \mathbf{U} \mathbf{g} 
             =\;  \boldsymbol\alpha^\text{T} (\mathbf{-A^{-1}}) \mathbf{g}.
\end{align*}
\end{proof}

\section{Minimum of an exponential and (generalized) Erlang distribution is Coxian} \label{A:minExpErlang}

The following is a known result, but is included here for completeness and for the benefit of readers unfamiliar with phase-type distributions. 

\begin{theorem} \label{Th:minExpErlang} Let random variable $Y$ be the minimum of an exponentially distributed random variable with rate $\mu$, and an independent generalized Erlang distributed (hypoexponential) random variable (i.e., the sum of $k$ independent exponentially distributed random variables) with parameters $\lambda_1$, $\ldots$, $\lambda_k$.  Then $Y$ obeys a Coxian distribution parameterized by 
	
	\begin{equation} \boldsymbol{\alpha} = \begin{bmatrix} 1 \\ 0 \\ \vdots \\ 0  \end{bmatrix}  \qquad \text{and} \qquad  \mathbf{G} = \begin{bmatrix} -\lambda_1 - \mu & \lambda_1 & 0 & \cdots & 0 \\
	0 & -\lambda_2 - \mu & \lambda_2  & \ddots & 0 \\
	\vdots & \ddots & \ddots  & \ddots & \ddots \\
	0 & 0 & \ddots &  -\lambda_{k-1} - \mu & \lambda_{k-1} \\
	0 & 0 & \cdots &  0 & -\lambda_k - \mu
	\end{bmatrix} . \end{equation}
	and the mean value of $Y$ is given by	
	\begin{equation}
	\text{E}(Y) =\;    \frac1{\lambda_1+\mu}  + \frac1{\lambda_2+\mu} \bigg(\frac{\lambda_1}{\lambda_1+\mu}\bigg) + \frac1{\lambda_3+\mu}\bigg(\frac{\lambda_1}{\lambda_1+\mu}\, \frac{\lambda_2}{\lambda_2+\mu}\bigg)  + \cdots + \frac1{\lambda_k+\mu}\, \prod_{i=1}^{k-1} \frac{\lambda_i}{\lambda_i+\mu}
	\end{equation}
	
\end{theorem}

\begin{proof}The generalized Erlang (hypoexponential) and exponential distributions can each be written as phase-type distributions, with parameterization
	
	\begin{equation}\boldsymbol{\alpha_1} = \begin{bmatrix} 1 \\ 0 \\ \vdots \\ 0  \end{bmatrix}  \qquad \text{and} \qquad \mathbf{A_1} = 
	\begin{bmatrix} -\lambda_1 & \lambda_1 & 0 & \cdots & 0 \\
					0 & -\lambda_2 & \lambda_2  & \ddots & 0 \\
				\vdots & \ddots & \ddots  & \ddots & \ddots \\
					0 & 0 & \ddots &  -\lambda_{k-1} & \lambda_{k-1} \\
					0 & 0 & \cdots &  0 & -\lambda_k
	\end{bmatrix} \end{equation}
	
	where $\mathbf{A_1}$ is a $k\times k$ matrix, and \begin{equation} \boldsymbol{\alpha_2}=[1] \qquad \text{and} \qquad \mathbf{A_2}=[-\mu].  \end{equation}

	By eq. \eqref{eq:minPT}, this minimum random variable $Y$ is also phase-type distributed, with parameters 
	\begin{subequations}  \begin{align} \boldsymbol{\alpha_1} \otimes \boldsymbol{\alpha_2} =&\; \boldsymbol{\alpha_1} \otimes [1] = \boldsymbol{\alpha}, \\
	\mathbf{A_1} \oplus \mathbf{A_2} =&\; \mathbf{A_1} \otimes [1] +  \mathbf{I}_{k\times k} \otimes [-\mu] = \mathbf{G}. \end{align} \end{subequations}
	
	By the structure of $\mathbf{G}$, one can be more specific and say that $Y$ has a Coxian distribution.  For this particular case, direct inspection confirms that 
	
	\begin{equation} -\mathbf{G}^{-1} = \begin{bmatrix} \frac1{\lambda_1+\mu} & \frac1{\lambda_2+\mu}\, \frac{\lambda_1}{\lambda_1+\mu} & \frac1{\lambda_3+\mu}\,\frac{\lambda_1}{\lambda_1+\mu}\,\frac{\lambda_2}{\lambda_2+\mu} & \cdots & \frac1{\lambda_k+\mu}\, \prod_{i=1}^{k-1} \frac{\lambda_i}{\lambda_i+\mu}\\
	0 & \frac1{\lambda_2+\mu} & \frac1{\lambda_3+\mu}\,\frac{\lambda_2}{\lambda_2+\mu} & \ddots & \frac1{\lambda_k+\mu}\, \prod_{i=2}^{k-1} \frac{\lambda_i}{\lambda_i+\mu} \\
	\vdots & \ddots & \ddots  & \ddots & \ddots \\
	0 & 0 & \ddots &  \frac1{\lambda_{k-1}+\mu} & \frac1{\lambda_{k}+\mu} \, \frac{\lambda_{k-1}}{\lambda_{k-1} + \mu}  \\
	0 & 0 & \cdots &  0 & \frac1{\lambda_k+\mu}
	\end{bmatrix} \end{equation}
	
	and thus, since $\boldsymbol{\alpha}$ is a 1 followed by $k-1$ zeros, it follows that the expected value of this distribution, E$(Y)=\boldsymbol{\alpha}^\text{T}(-\mathbf{G}^{-1})\mathbf{1}$, is the sum of the entries in the first row of $-\mathbf{G}^{-1}$.
	
\end{proof}

\end{appendices}

\clearpage
\printbibliography

\end{document}